\documentclass[11pt]{article}
\usepackage[a4paper, top=2.5cm, bottom=2.5cm, left=2.5cm, right=2.5cm]
{geometry}

\usepackage[utf8]{inputenc} 
\usepackage[T1]{fontenc}    
\usepackage{hyperref}       
\usepackage{url}            
\usepackage{booktabs}       
\usepackage{amsfonts}       
\usepackage{nicefrac}       
\usepackage{microtype}      
\usepackage{xcolor}         

\usepackage{graphicx}
\usepackage{amsmath}
\usepackage{amsthm}
\usepackage{amsfonts}
\usepackage{amssymb}
\usepackage{natbib}
\usepackage{makecell}
\usepackage[linesnumbered,algoruled,boxed,lined]{algorithm2e}
\usepackage{xcolor}
\usepackage{subcaption}
\usepackage{soul}
\usepackage{multirow}
\usepackage{hyperref}
\usepackage{booktabs}

\usepackage{tikz}
\usetikzlibrary{arrows.meta}
\usepackage{algpseudocode,float}
\usepackage{lipsum}
\usetikzlibrary{positioning,decorations.pathreplacing}
\usetikzlibrary{calc}
\usetikzlibrary{shapes,backgrounds}
\usepackage{makecell}

\newtheorem{definition}{Definition}
\newtheorem{proposition}{Proposition}

\newtheorem{theorem}{Theorem}

\newtheorem{remark}{Remark}
\newtheorem{lemma}{Lemma}

\usepackage{makecell}
\usepackage{enumitem}
\usepackage{microtype}
\usepackage{multirow}
\usepackage{bm}
\usepackage{bigstrut}
\usepackage{subcaption}
\usepackage{xcolor}

\title{Centralized Group Equitability and Individual Envy-Freeness in the Allocation of
Indivisible Items}

\author{
Ying Wang$^{1}$ \hspace{15pt} 
Jiaqian Li$^{2}$  \hspace{15pt}
Tianze Wei$^{3}$ \thanks{Corresponding author.} \hspace{15pt} 
Hau Chan$^{4}$  \hspace{15pt}
Minming Li$^{3}$\\
 $^{1}$Columbia University \hspace{15pt}
$^{2}$Boston University \hspace{15pt}\\
$^{3}$City University of Hong Kong \hspace{15pt}
$^{4}$University of Nebraska--Lincoln\\
\small \texttt{yw4360@columbia.edu,
jiaqian@bu.edu,
t.z.wei-8@my.cityu.edu.hk}\\
\small \texttt{hchan3@unl.edu,
minming.li@cityu.edu.hk}
}

\date{}

\begin{document}
\maketitle
\begin{abstract}
We study the fair allocation of indivisible items for groups of agents from the perspectives of the agents and a centralized allocator. 
In our setting, the centralized allocator is interested in ensuring the allocation is fair among the groups and between agents. This setting applies to many real-world scenarios, including when a school administrator wants to allocate resources (e.g., office spaces and supplies) to staff members in departments and when a city council allocates limited housing units to various families in need across different communities. 
To ensure fair allocation between agents, we consider the classical \textit{envy-freeness} (EF) notion.
To ensure fairness among the groups, we define the notion of \textit{centralized group equitability} (CGEQ) to capture the fairness for the groups from the allocator's perspective. 
Because an EF or CGEQ allocation does not always exist in general, we consider their corresponding natural relaxations of \textit{envy-freeness to one item} (EF1) and \textit{centralized group equitability up to one item} (CGEQ1). 
For different classes of valuation functions of the agents and the centralized allocator, we show that allocations satisfying both EF1 and CGEQ1 always exist and design efficient algorithms to compute these allocations.
We also consider the \emph{centralized group maximin share} (CGMMS) from the centralized allocator's perspective as a group-level fairness objective with EF1 for agents and present several results.
\end{abstract}

\section{Introduction}

Fair division of indivisible items often deals with fairly allocating a set of (discrete or indivisible) items to a set of agents who have preferences over the items. 
Due to both practical and theoretical interests, fair division of indivisible items has received considerable attention in various research communities, such as economics, mathematics, and computer science, for much of the past century \citep{moulin2004fair,DBLP:journals/ai/AmanatidisABFLMVW23}. 
In practice, fair division of indivisible items has many real-world applications ranging from course allocation (i.e., for allocating schedules of courses to students) to goods division (i.e., dividing artworks or furniture among individuals), in which the Course Match \citep{budish2017course} mechanism and the Spliddit \citep{goldman2015spliddit} online platform have been developed, to provide fair allocations subject to agent preferences and appropriate fairness notions (e.g., envy-freeness and their relaxations) for the respective applications. 
In theory, fair division of indivisible items has led to the development of numerous notions such as envy-freeness up to one/any item, proportionality up to one/any item, and maximin share fair for quantifying fairness, algorithms (e.g., round-robin or envy-graph procedure \citep{lipton2004approximately}) for providing (approximately) fair allocations, and techniques for (partially) characterizing the existence of fair allocation \citep{DBLP:journals/ai/AmanatidisABFLMVW23}. 


\medskip
\noindent
\textbf{Allocator's Preference.} 
A main drawback of existing studies of fair division of indivisible items is the lack of consideration from the allocator's perspective, who is responsible for implementing the allocation and has preferences on how the items should be allocated to the set of agents \citep{bu2023fair}. 
As a result, \citep{bu2023fair} initiated the study of fair division of indivisible items from both the allocator’s and agents' perspectives, where each agent has a valuation preference over the items, and the allocator has a separate valuation preference for each agent over the items (specifying their internal values for the agent receiving the items). 
Moreover, they focus on allocations that satisfy relaxations of envy-freeness between agents under both the allocator's valuation for each agent and agent valuations. 


As a motivating example, \citep{bu2023fair} discussed the situations where the government (as the allocator) needs to distribute education resources (e.g., funding and staff members) to different schools (as agents) in which the schools and the government have separate preferences over the education resources based on their needs and macroeconomic policy for schools, respectively. In addition, \citep{bu2023fair} provided examples where a company allocates resources to different departments, an advisor allocates tasks or projects to students, and conference organizers allocate papers to reviewers that require the consideration of both the allocator's and agents' preferences. 


\medskip
\noindent
\textbf{Our Study: Centralized Allocator's Preference.} 
Building on the work of \citep{bu2023fair}, we introduce a \emph{centralized} allocator who is interested in ensuring fairness at a group level, where each agent naturally belongs to different predefined groups in the fair allocation.
For instance, building on the above-mentioned example, a school administrator, tasked with allocating limited resources (e.g., office spaces and supplies) to staff members from departments within the school \citep{perez2022role}, needs to ensure that the allocation is also fair at the department (group) level. 
A city council, tasked with allocating limited housing units to various neighborhoods in need across different communities \citep{gray1976selection}, needs to ensure the allocation is fair with respect to different communities. 
Finally, a government distributing resources to different schools needs to ensure that the allocation is fair with respect to the schools. 
Therefore, in this paper, our goal is to explore the fair division of indivisible items, which not only provides fairness for the agents but also guarantees fairness for the centralized allocator.

\begin{table}[tb]
    \centering 
    \begin{tabular}{cccc}
    \toprule
       \makecell[c]{Centralized \\ allocator's valuation} &\makecell[c]{Agents'\\ valuations}& \makecell[c]{EF1+CGEQ1}\\     
    \midrule
    Arbitrary &Identical  & \makecell[c]{$\checkmark$ (Poly time) (Thm \ref{the: idential_valuations})} \\
    \midrule
    \multicolumn{2}{c}{Ordered} & \makecell[c]{$\checkmark$  (Poly time) (Thm \ref{the: ordered_instance}) }\\
    \midrule
     Binary & \makecell[c]{Arbitrary}  & \makecell[c]{$\checkmark$ (Poly time) (Thm \ref{the: binary_ef1_CGEQ1}) }\\
    \bottomrule
    \end{tabular}
    \caption{Summary of our main results.}
    \label{tab:contribution}
\end{table}


\subsection{Our Contribution}

We study the fair division of indivisible items for groups of agents from the perspectives of the agents and the centralized allocator. 
Each agent belongs to one of the groups (e.g., based on their associations) and has an additive valuation function over the items. 
The centralized allocator has a common additive valuation function indicating their values for the items measured in standardized units (e.g., investment value, monetary amount, and space). 

To ensure fair allocation among agents, we consider the classical envy-freeness (EF) notion. 
To ensure fairness among the groups, we define the notion of centralized group equitability (CGEQ) to capture the fairness for the groups from the centralized allocator’s perspective that compares the weighted proportion of values received by each group. 
Because an EF or CGEQ allocation does not always exist---and in fact checking for a pure CGEQ outcome is computationally strong NP-hard via a reduction from the 3-partition problem, and deciding whether an allocation satisfies EF+CGEQ, EF+CGEQ1, or EF1+CGEQ is likewise NP-hard---we therefore consider their natural relaxations of envy-freeness up to one item (EF1) and centralized group equitability up to one item (CGEQ1). Following the idea from \citep{bu2023fair}, we strive to answer the following questions. 

\begin{quote}
\emph{Under which conditions can we guarantee the existence of EF1+CGEQ1 allocations? If so, can we design algorithms to compute them efficiently? } 
\end{quote} 

To address the above questions, we examine different classes of valuation functions of the agents/allocator. 
The presence of the centralized allocator introduces a fundamental shift in both the fairness notions and the algorithmic challenges involved. The techniques we develop, though sometimes inspired by classic methods such as round-robin, are nontrivial extensions that integrate allocator-aware priorities and group-level proportionality.
Specifically, our key contributions are as follows (summarized in Table~\ref{tab:contribution}):
\begin{itemize} \item When each agent has an identical valuation function, even though agents are indistinguishable in terms of preferences, the allocator's independent valuation introduces nontrivial global constraints. Our \textbf{DM Algorithm} (Algorithm~\ref{alg: staggered matching}) constructs a temporary allocation satisfying EF1 for both agents and allocator, using a match-based process guided by the allocator’s preferences. Then, it reallocates bundles to achieve CGEQ1 while preserving agent-level EF1. This dynamic reassignment highlights the allocator's role in determining group-level equity even under agent homogeneity.
\item When all agents and the allocator share the same ordinal ranking of items, we propose the \textbf{SPS Algorithm} (Algorithm~\ref{alg: ido_case}) that simultaneously addresses item distribution across groups and within groups. Despite the aligned ordering, the absolute values may differ significantly under separate allocator's and agents' valuation functions. The allocator’s proportional fairness criteria---considering each group’s value-to-size ratio---guide item assignment to balance CGEQ1 and agent-level EF1. This synchronization of dual fairness notions introduces dependencies absent in standard ordinal settings.
\item When the allocator classifies items into two types. Our \textbf{CD$^2$P Algorithm} (Algorithm~\ref{alg: Binary_allocator}) innovatively extends Round-Robin by introducing a reverse RR phase. This two-phase procedure ensures EF1 for agents while achieving CGEQ1, and cannot be reduced to standard round-robin without losing group fairness. 
\item We also propose the \textbf{Centralized Group Maximin Share} (CGMMS) as an allocator-centric optimization benchmark and seek allocations that satisfy CGMMS for the centralized allocator and EF1 for agents.
\end{itemize}

The remainder of the paper is organized as follows. In Section \ref{sec: preliminary}, we formally define the notations and fairness notions considered in our paper.
In Sections \ref{sec:: identical valuations}, \ref{sec:: ordered instance}, and \ref{sec: binary_allocator}, we study EF1+CGEQ1 allocations in identical valuations, ordered valuations, and binary valuations settings, respectively. 
In Section \ref{sec: efficiency}, we discuss EF1+CGMMS allocations.
In Section \ref{sec::conlusion}, we conclude the paper and provide future research directions.
Due to space constraints, we refer readers to the appendix for the missing proofs.

\subsection{Related Work}
There is an extensive line of work in the fair division of indivisible items. 
We refer readers to the survey \citep{DBLP:journals/ai/AmanatidisABFLMVW23} for an overview. 
Below, we review studies focusing on allocations that consider group fairness and fairness from the agents' and allocator's perspectives. 

\medskip
\noindent
\textbf{Fairness from the Agents' and Allocator's Perspectives.} 
As discussed earlier, the most relevant work is \citep{bu2023fair}, where they initiated the study of fair division of indivisible items from the perspectives of the agents and the allocator. 
There are significant differences between our paper and their model. More specifically,
\begin{itemize}
    \item \citep{bu2023fair}'s ``two-layer" paradigm requires a single allocation to satisfy two distinct sets of valuation functions for individuals (agents vs. allocator). Our model adds yet another layer on top of this: a group structure in which the allocator cares only about each group's aggregate share.
    \item When we let each group have exactly one agent, our model degenerates to their model, where the allocator has the identical valuation function. Additionally, due to the group setting in our paper, the techniques used in their paper, including envy-cycle elimination and round-robin, cannot be directly applied to our model.
    \item  A concrete example of technology transfer failure arises with \citep{bu2023fair}'s doubly-EF1 algorithm, which hinges on cycle elimination. If one naively ports that step to the group model, moving an item that happens to be a group's extreme-value good immediately violates CGEQ1. To safeguard the group-level guarantees, we must introduce a tailored combination of bundle duplication and a quota inequality, which forms the core of our Lemma \ref{lem:staggered matching-2}.
\end{itemize} 

\citep{flammini2025fair} also introduced the second valuation functions in the setting, but they focused on the loss in efficiency with respect to the second valuation functions while pursuing the fairness among agents, where the target is different from our paper.
 
\medskip
\noindent
\textbf{Group Fairness.}
Existing studies have examined group-fair division of indivisible items from only the agent perspective. 
Some works focus on the predefined group.
For example, \citep{aleksandrov2018group} defined group envy-freeness and group Pareto optimality, and studied the price of group envy-freeness.
\citep{benabbou2019fairness} considered the fair matching among different groups where each agent can pick at most one item. They studied the fairness criteria named typewise envy-freeness up to one item (TEF1), and showed that when agents have binary valuations, TEF1 allocations can be computed in polynomial time.
\citep{feige2022allocations} studied the notion of group proportional share fairness and group any price share fairness in different groups that may have different structures like laminar. \citep{manurangsi2024ordinal} studied the ordinal maximin share fairness among groups.
There are other works that did not consider the predefined group. For example, \citep{conitzer2019group} studied the group fairness among agents, where they considered any partition of agents, and showed that local optimal Nash welfare allocations satisfy two different relaxations of group fairness that they defined.
Later, \citep{aziz2021almost} extended it to the setting where items include both goods and chores.
The survey of \citep{DBLP:journals/ai/AmanatidisABFLMVW23} offers a comprehensive view of recent progress and open problems in this field.

The most related setting to the proposed study is the work of \citep{scarlett2023one}, where they studied the compatibility of individual envy-freeness and group envy-freeness from the agent perspective only. 
Moreover, they did not consider the centralized allocator and defined the group utility based on the agent's valuation function instead of the centralized allocator's valuation. 






\section{Preliminaries}
\label{sec: preliminary}




In this section, we present notations and fairness notions for the considered setting of fair division of indivisible items with a set of agents and a centralized allocator.





\subsection{Notations}

For $r \in \mathbb{N}$, let $[r] = \{1,2,\dots,r\}$. 
Let $\mathcal{O}=\{o_1,o_2,\dots,o_m\}$ be a set of $m$ indivisible items, and $\mathcal{N}=[n]$ be a set of $n$ agents. 
The set of $n$ agents is partitioned into $k \in \mathbb{N}$ groups denoted by $\mathcal{G} = (G_1, \ldots, G_k)$. Each agent belongs to exactly one group, i.e., $G_p \cap G_q = \emptyset$ for any $p,  q \in [k]$. Additionally, our setting includes a \textit{centralized} allocator.


Each agent $i \in \mathcal{N}$ has an additive valuation function $v_i: 2^{\mathcal{O}} \rightarrow \mathbb{R}_{\geq 0}$, i.e., for any $S \subseteq \mathcal{O}$, $v_i(S) = \sum_{o \in S}v_i(\{o\})$. Specifically, we assume that $v_i(\emptyset) = 0$. 
The centralized allocator has her own preferences and is endowed with an additive valuation function $u$: $2^\mathcal{O}$ $\rightarrow \mathbb{R}_{\geq 0}$, i.e., for any $S \subseteq \mathcal{O}$, $u(S) = \sum_{o \in S}u(\{o\})$ indicating their values for the items measured in standardized units (e.g., investment value, monetary amount, and space). 
Additionally, we assume that $u(\emptyset) = 0$.

For simplicity, we use $v_i(o)$ and $u(o)$ instead of $v_i(\{o\})$ and $u(\{o\})$, respectively.
Let $\Pi(n,\mathcal{O})$ denote all $n$-partitions of $\mathcal{O}$.
An allocation $\mathcal{A} = (A_1, \ldots, A_n) \in \Pi(n, \mathcal{O})$ is an $n$-partition of $\mathcal{O}$ among $n$ agents, where $A_i$ is the bundle allocated to agent $i$. We have $\bigcup_{i \in \mathcal{N}}A_{i} = \mathcal{O}$ and $A_{i} \cap A_{j} = \emptyset$ for any two agents $i \neq j$.
A fair allocation instance is denoted as $\mathcal{I} = \langle \mathcal{O}, \mathcal{N}, \mathcal{G}, \boldsymbol{v}, u \rangle$, where $\boldsymbol{v} = (v_1, \ldots, v_n)$.
Next, we provide the formal definitions of the fairness notions for the agents and the centralized allocator.

\subsection{Fairness and Efficiency Notions}

To ensure fairness among agents, we consider the classical envy-freeness (EF) notion from the agent perspective.

\begin{definition}[Envy-Freeness] An allocation $\mathcal {A}$ is envy-free (EF), if for any two distinct agents $i,j \in \mathcal{N}$, we have $v_i(A_i) \geq v_i(A_j)$.
\end{definition}

However, EF allocations do not always exist. 
Therefore, we consider a natural and commonly studied relaxation of EF, named envy-free up to one item.

\begin{definition}[Envy-Freeness up to One Item] An allocation $\mathcal {A}$ is 
envy-free up to one item (EF1) if, for any two agents $i,j \in \mathcal{N}$, $v_i(A_i) \geq v_i(A_j \setminus \{o\})$ holds for some $o \in A_j$.
\end{definition}

Next, we introduce our fairness notion from the centralized allocator's perspective, which is called centralized group equitability (CGEQ). 

\begin{definition}[Centralized Group Equitability]
     An allocation $\mathcal {A}$ is 
     centralized group equitable (CGEQ) if for any two groups $G_p, G_q \in \mathcal{G}$, $\frac{u( \bigcup_{i \in G_p}A_i)}{\lvert G_p \rvert } = \frac{u(\bigcup_{j \in G_q}A_j)}{\lvert G_q \rvert }$ holds. 
\end{definition}



This definition reflects an \emph{equitable view from the centralized allocator's perspective}, where the utility function $u(\cdot)$ represents the authority’s valuation over bundles allocated to different groups. CGEQ thus requires that each group receives, on average, the same utility as any other group, when judged by the central authority. 

\begin{remark}
We intentionally label this condition ``equitability" (CGEQ) instead of ``envy-freeness" because the focus is on equalizing welfare rather than eliminating envy. In other words, the central planner is not checking if group $G_p$ prefers group $G_q$'s bundle (which would imply a notion of group envy). Rather, she is ensuring that her common yardstick values each group's allocation equally on a per-agent basis. This perspective aligns with \citet{conitzer2019group}'s group fairness structure, but we adapt it to heterogeneous valuations by following the framework of \citet{bu2023fair}, where the allocator's valuation may differ from individual agents'. By synthesizing these approaches, we obtain a group-level fairness notion judged by the central authority's values.
\end{remark}
 

The exact fairness conditions (e.g., EF) are often too strong to satisfy with indivisible items. Here, CGEQ allocations do not always exist.
Consider an instance with only one indivisible item and two groups of equal size: any allocation gives utility to only one group, making CGEQ impossible. 
Therefore, we propose the following relaxation notion.



\begin{definition}[Centralized Group Equitability up to One Item]
     An allocation $\mathcal {A}$ is said to be centralized group equitable up to one item (CGEQ1) if, for any two groups $G_p, G_q \in \mathcal{G}$, $\frac{u(\bigcup_{i \in G_p}A_i)}{\lvert G_p \rvert } \geq \frac{u(\bigcup_{j \in G_q}A_j \setminus \{o\})}{\lvert G_q \rvert }$ holds for some item $o \in \bigcup_{j \in G_q}A_j$.
\end{definition}

The relaxed version, CGEQ1 (Centralized Group Equitability up to One Item), follows the spirit of the classic EF1 relaxation \cite{lipton2004approximately} but in terms of equitability: each group's average utility is almost equal, possibly differing only by the value of a single item. We use the term CGEQ1 (not CGEF1) to emphasize this distinction.


We are interested in allocations that satisfy EF1 from the agent's perspective and CGEQ1 from the centralized allocator's perspective.
In particular, we study computing EF1+CGEQ1 allocations in various scenarios, focusing on representative valuation functions of the agents and the centralized allocator that have each been the subject of extensive recent study in the fair‐division literature---Identical Valuation: \cite{barman2025fair, birmpas2024fair, lang2024fair, babichenko2024fair}; Ordered Valuation: \cite{neoh2025online, garg2025existence, garg2025constant, garg2025improved};  Binary Valuation: \cite{chandramouleeswaran2025fair, bismuth2024fair, deligkas2024complexity,  gorantla2023fair}. 
If the size of any group is one, i.e., $\lvert G_p \rvert = 1$ for any $G_p \in \mathcal{G}$, CGEQ1 degenerates into EF1.
In this case, our setting reduces to a special case in \cite{bu2023fair}, where they showed that an EF1+EF1 allocation can be computed in polynomial time.
Thus, we consider the case where the size of some group is not one in the following sections.

\section{Identical Valuations}\label{sec:: identical valuations}

In this section, we first consider the case where the valuations of each agent and the centralized allocator are the same, i.e., $v_1 = \cdots = v_n = u$.  In that case, our setting degenerates to that in \cite{scarlett2023one}, where they study the combination of group and individual fairness and define the group utility based on one set of valuations (i.e., the agents' valuations).





Next, we consider the case where only the valuation of each agent is the same, i.e., $v_1 = \cdots=v_n$, while the centralized allocator's valuation is arbitrary.
The motivating example would be that in textbook/course-seat allocations, agents have very similar needs for standardized items, while the allocator pursues distributional goals (access, equity).

For simplicity, let $v$ denote the valuation of each agent. In that case, the algorithm in \cite{scarlett2023one} does not apply directly since its technique heavily depends on the fact that the group utility is computed by using the agents' valuations.
However, the centralized allocator's valuation can be different from the agents' valuations in our setting.
Therefore, we propose a new algorithm named Draft-and-Match (DM), as described in Algorithm \ref{alg: staggered matching}.
The whole algorithm includes two phases. In Phase 1, the intuition of our algorithm is to partition the items into a temporary allocation $\mathcal{A}^{\prime}$ by allocating the item that has the highest value to the agent whose bundle has the lowest value in each iteration. This approach ensures that the allocation is envy-free up to one item (EF1) with respect to both the agents' and the centralized allocator's valuation functions.
In Phase 2, given the temporary allocation, we then follow specific rules to match and reallocate these bundles to the agents, which ensures the final allocation is CGEQ1 for the centralized allocator. 
Since we consider identical agents' valuations, this reallocation does not violate the EF1 property of each agent in $\mathcal{A^{\prime}}$, and the returned allocation is EF1+CGEQ1.






\begin{theorem}
\label{the: idential_valuations}
    Given an instance where the valuation of each agent is the same,  the Draft-and-Match Algorithm (Algorithm \ref{alg: staggered matching})  
    computes an EF1+CGEQ1 allocation in polynomial time.
\end{theorem}

\begin{algorithm}[!htb]
 \caption{Draft-and-Match (DM)}
 \label{alg: staggered matching}
  \KwIn{An instance $\mathcal{I} = \langle \mathcal{O}, \mathcal{N}, \mathcal{G}, \boldsymbol{v}, u \rangle$ with identical agents' valuation functions}
  \KwOut{An EF1+CGEQ1 allocation $\mathcal{A}$}
  \SetKwData{Left}{left}\SetKwData{This}{this}\SetKwData{Up}{up}
  \SetKwFunction{Union}{union}\SetKwFunction{FindCompress}{findcompress}
  \BlankLine       
  Let $\mathcal{A}'= (\emptyset, \ldots, \emptyset)$  and $\mathcal{A}=(\emptyset, \ldots, \emptyset)$;\\

  {\color{gray}  - - - - -  Phase 1: Partition $\mathcal{O}$ into allocation $\mathcal{A}^{\prime}$ - - - - - } 
  
  Add $n- (m \mod n)$ dummy items where each agent and the centralized allocator have the valuation of zero to $\mathcal{O}$;\\
   Let $\mathcal{O}_s$ be the array of sorted goods with respect to $u$ in non-increasing order;\\
 
    \While{$\mathcal{O}_s \neq \emptyset$ }{
    Let $\mathcal{N}^{\prime} = \mathcal{N}$;\\
    Let $\mathcal{O}_n$ be the first $n$ items in $\mathcal{O}_s$ and $\mathcal{O}_s \leftarrow \mathcal{O}_s \setminus \mathcal{O}_n$;\\
     
     \While{$\mathcal{O}_n \neq \emptyset$ }{   
     $A_{i_{min}}^{'} \leftarrow A_{i_{min}}^{'} \cup \{o_{max}\}$, where $i_{min} \in \arg \min_{i \in \mathcal{N}^{\prime}} v(A_{i}^{\prime})$ and $o_{max} \in \arg \max_{o \in \mathcal{O}_n} v(o)$ (breaking ties arbitrarily);
    
    
      $\mathcal{N}^{\prime} = \mathcal{N}^{\prime} \setminus \{i_{min}\}$ and $\mathcal{O}_n \leftarrow \mathcal{O}_n \setminus \{o_{max}\}$;
    }
    }
   

    


 {\color{gray}   - - -    Phase 2: Match the bundles in $\mathcal{A}^{\prime}$ to agents - - -} 
 
     Assume that the groups are in non-decreasing order of size, i.e., $\lvert G_1 \lvert \leq \ldots \leq \lvert G_k \rvert$;\\

    $t_p \leftarrow 0$,  $ \forall p\in [k]$, $t_p$ is the total number of times that $G_p$ has been picked
so far;
     
     \While {$\mathcal{A}^{'} \neq (\emptyset, \ldots, \emptyset)$}{

     \If{$\exists t_{p} = 0 $}{
       $ p^* \leftarrow \min \{p \vert p \in [k] ~\textit{and}~ t_p = 0\}$;\\   
     }
     \Else{

     $p^* \leftarrow \arg \min_{p\in [k]}{\frac{t_p}{|G_p|}}$ (breaking ties by selecting $p$ that reaches $\min_{p\in [k]}{\frac{t_p}{|G_p|}}$ the latest);\\ 
     } 

     Arbitrarily choose one agent $i^* \in G_{p^{*}}$, where $A_{i^*} = \emptyset$;
     
     $A_{i^*} \leftarrow  A'_{i_{max}}$, where $A'_{i_{max}} \in \arg \max_{A'_{i} \in \mathcal{A}'} u(A_{i}^{'})$;\\
     
     $t_{p^*} \leftarrow t_{p^*}+1$;\\
     $A_{i_{max}}^{'} \leftarrow \emptyset $; 
    }    
    \Return $\mathcal{A}$ 
     



\end{algorithm}

Before proving Theorem \ref{the: idential_valuations}, we first give the following lemmas.
In Lemma \ref{lem:staggered matching-1}, we show that the temporary allocation $\mathcal{A}^{\prime}$ is EF1 with respect to both the agents' and the centralized allocator's valuation functions.
Then in Lemma \ref{lem:staggered matching-2}, we show that after the reallocation of bundles in $\mathcal{A}^{\prime}$, the returned allocation $\mathcal{A}$ is CGEQ1.






\begin{lemma}\label{lem:staggered matching-1}

    The allocation $\mathcal{A}'$ (computed in Phase 1) in Algorithm \ref{alg: staggered matching} is EF1 with respect to both the agents' and the centralized allocator's valuation functions. 
That is, for any $i, j \in [n]$, we have $v(A'_i) \geq v(A'_j \setminus \{o\})$ and $u(A'_i) \geq u(A'_j \setminus \{o^{\prime}\})$ for some $o,o^{\prime}\in A'_j$. 
\end{lemma}

\begin{proof}
    We use mathematical induction to show that for each agent $i \in \mathcal{N}$ and the centralized allocator, 
    the temporary allocation  $\mathcal{A^{\prime}}$ is EF1 with respect to their valuation functions.
    Since all agents have the same valuation function, no envy cycle occurs in any allocation. 
    Hence, there is no need to eliminate envy cycles during the iteration. 
    This claim also holds for the centralized allocator.
    
    For the base case, it is clear that the empty bundle is EF1 with respect to the agents' valuation functions and the centralized allocator's valuation function.
    
    Next, for the induction step, we assume that in the $k$th iteration, any bundle in the partial allocation $\mathcal{A}^{\prime(k)} = (A_1^{\prime(k)}, \ldots, A_n^{\prime(k)})$ is EF1 with respect to the agents' valuation functions and the centralized allocator's valuation function, i.e., for any two bundles $A_i^{\prime(k)}$ and $A_{j}^{\prime(k)}$, both $v(A_{i}^{\prime(k)}) \geq v(A_{j}^{\prime (k)} \setminus \{o\})$ and $u(A_{i}^{\prime(k)}) \geq u(A_{j}^{\prime (k)} \setminus \{o\})$ hold for some $o \in A_{j}^{\prime (k)}$. We now show that this property continues to hold after the $(k+1)$th iteration.
    
    Fix two arbitrary agents $i, j \in \mathcal{N}$.
    Let $o^{(k+1)}_{i}$ and $o^{(k+1)}_{j}$ denote the item allocated to agent $i$ and $j$, respectively.
    Without loss of generality, assume that agent $j$ picks $o^{(k+1)}_{j}$ first, which implies that $v(A_i^{\prime(k)})\geq v(A_j^{\prime(k)})$, $v(o^{(k+1)}_{j}) \geq v(o^{(k+1)}_{i})$, and $u(o^{(k+1)}_{j}) \geq u(o^{(k+1)}_{i})$.
    For agent $i$, we have
    \begin{equation*}
        \begin{aligned}
         v(A_{i}^{\prime(k+1)})  \geq v(A_{i}^{\prime(k)}) \geq v(A_{j}^{\prime(k)}) = v(A_{j}^{\prime(k+1)} \setminus \{o^{k+1}_{j}\}),   
        \end{aligned}
    \end{equation*}
    and for agent $j$, we have
    \begin{equation*}
        \begin{aligned}
     v(A_{j}^{\prime(k+1)}) & = v(A_{j}^{\prime(k)}) +v (o_{j}^{(k+1)}) \\
     &\geq v(A_{i}^{\prime(k)} \setminus \{o\}) + v(o_{i}^{(k+1)}) = v(A_{i}^{\prime(k+1)} \setminus \{o\})  
        \end{aligned}
    \end{equation*} 
    for some $o  \in A_{i}^{\prime(k)}$.
    Therefore, after the $(k+1)$th iteration, any bundle in the partial allocation is still EF1 to each agent's valuation function.
    Since the items are indexed in non-increasing order of the centralized allocator's valuation, and the first $n$ items from the remaining items are allocated to agents in each iteration, we get $u(o^{(\ell)}_{j}) \geq u(o^{(\ell +1)}_{i})$ for any two agents $i, j \in \mathcal{N}$, and any $\ell \in [k]$.

    Then, for the centralized allocator, it holds that $\forall i,j \in \mathcal{N}$
    \begin{equation*}   
    u(A_{i}^{\prime(k+1)}) = \sum_{s = 1}^{k+1}u(o^{(s)}_{i}) 
      \geq \sum_{s=2}^{k+1}u(o^{(s)}_{j}) = u(A_{j}^{\prime(k+1)} \setminus \{o^{(1)}_{j}\}) 
    \end{equation*}
    
    Thus, after the $(k+1)$th iteration, any bundle in the partial allocation is still EF1 to the centralized allocator's valuation function.
    This completes the induction and establishes the correctness of our proof.
\end{proof}

Since the agents' valuation functions are identical, it allows us to rearrange the bundles without violating fairness at the individual level. Leveraging this flexibility, we aim to rearrange the bundles to meet the group-level fairness (CGEQ1) while preserving EF1 for the agents.

\begin{lemma}
\label{lem:staggered matching-2}
    In Algorithm \ref{alg: staggered matching}, the returned allocation $\mathcal{A}$ is CGEQ1.
\end{lemma}

\begin{proof}
It suffices to show that for any two groups $G_p$ and $G_q$, $\frac{u(\bigcup_{i \in G_p}A_i)}{\lvert G_p \rvert } \geq \frac{u(\bigcup_{j \in G_q}A_j \setminus \{o\})}{\lvert G_q \rvert }$ holds for some item $o \in \bigcup_{j \in G_q}A_j$. 
Without loss of generality, we assume that $|G_p| \leq |G_q|$. 


Let $B_{p}^1, B_{p}^2, \ldots, B_{p}^{|G_p|}$ represent the bundles received by agents in $G_p$ during the re-allocation process, where $B^i$ is the bundle received by the agent in $G_p$ during the $i$-th allocation to $G_p$. Similarly, let $B_{q}^1, B_{q}^2, \ldots, B_{q}^{|G_q|}$ represent the bundles received by agents in $G_q$ during their respective allocations. Between the allocation of bundles $B_{p}^i$ and $B_{p}^{i+1}$ to $G_p$, note that some agents in $G_q$ may also receive bundles. Let $f_i$ denote the number of bundles received by agents in $G_q$ between the allocation of $B_{p}^i$ and $B_{p}^{i+1}$ for $1 \leq i \leq |G_p| - 1$. We define $f_0$ as the number of bundles received by $G_q$ before $B_{p}^1$ is allocated to $G_p$, and $f_{|G_p|}$ as the number of bundles received by $G_q$ after $B_{p}^{|G_p|}$ is allocated. We have $\sum_{i = 0}^{|G_p|} f_i = |G_q|$.

We first show that $G_p$ will not envy $G_q$, and the proof of the other direction ($G_q$ will not envy $G_p$) is similar. 

To prove the desired inequality 
$$
\frac{u\left(\bigcup_{i \in G_p} A_i\right)}{|G_p|} \geq \frac{u\left(\bigcup_{j \in G_q} A_j \setminus \{o\}\right)}{|G_q|}
$$
for some item $o \in \bigcup_{j \in G_q} A_j$, it is sufficient to show that 
\begin{equation}\label{lemmaineq1}
    u\left(\bigcup_{i \in G_p} A_i\right) \cdot |G_q| \geq u\left(\bigcup_{j \in G_q} A_j \setminus \{o\}\right) \cdot |G_p|.
\end{equation}

Given the EF1 property established previously, for any bundle, removing its most valuable item results in a lower value than any other bundle. Thus, we aim to prove a stronger condition: we always remove the most valuable item from the first bundle received by $G_q$ and then compare the average values of the  bundles received by both groups. Equivalently, we construct the following strategy to compare the two sides of the inequality:
\begin{itemize}
    \item For each bundle received by $G_p$, duplicate it $|G_q|$ times.
   \item For each bundle received by $G_q$, duplicate it $|G_p|$ times.
\end{itemize}

This results in a total of $|G_p| \times |G_q|$ bundles for each group. We will now compare the total value of these duplicated bundles.

We observe that after $G_p$ receives a bundle, $G_q$ will receive at least one bundle before the next time $G_p$ gets another one. This occurs because $|G_p|$ is relatively smaller than $|G_q|$. 

At the moment when $G_p$ is selected to receive a new bundle, the allocation ratio $\frac{t_p}{|G_p|}$ must be smaller than or equal to $\frac{t_q}{|G_q|}$. We will now prove that this difference is at most $\frac{1}{|G_q|}$. This is because $G_p$ cannot receive two bundles consecutively, which means that the last bundle before $G_p$'s new allocation must have been received by $G_q$. Since $G_q$ received this last bundle, its allocation ratio $\frac{t_q}{|G_q|}$ was smaller than or equal to $\frac{t_p}{|G_p|}$ right before receiving it (otherwise, $G_p$ would have been chosen to receive the bundle). After $G_q$ receives the bundle, its ratio increases by exactly $\frac{1}{|G_q|}$. Therefore, we have 
$$
\frac{\sum_{i=0}^{k} f_i}{|G_q|} \leq \frac{k}{|G_p|} + \frac{1}{|G_q|}.
$$
Rearranging this inequality, we obtain
\begin{equation}\label{Lemma2ineq}
\left(\sum_{i=0}^{k} f_i \right) \cdot |G_p| \leq k \cdot |G_q| + |G_p|.
\end{equation}

Next, consider the order of the duplicated bundles for both groups. For $G_q$, the value of the duplicated bundles is arranged in descending order as follows: first $B_{q}^{2}, B_{q}^{3}, \dots, B_{q}^{|G_q|}$, each repeated $|G_p|$ times. Finally, $B_{q}^{1}$, with its most valuable item removed, is repeated $|G_p|$ times. Similarly, for $G_p$, the value of its duplicated bundles is arranged in descending order as $B_{p}^1, B_{p}^2, \dots, B_{p}^{|G_p|}$, each repeated $|G_q|$ times.

We can perform a pairwise comparison of the total values of corresponding duplicated bundles from both groups. Specifically, we compare the $i$th highest value bundle in $G_p$ with the $i$th highest value bundle in $G_q$. For the bundles received by $G_p$, consider the $((k-1)|G_q| + 1)$th to $(k \cdot |G_q|)$th bundles, for $1 \leq k \leq |G_p|$. These bundles have the value of $u(B_{p}^k)$, which is smaller than the values of at most 
\[
\sum_{i = 0}^{k - 1} f_i \cdot |G_p| - |G_p|
\]
bundles received by $G_q$ (in the duplicated scenario). By Inequality (\ref{Lemma2ineq}), $
\sum_{i = 0}^{k - 1} f_i \cdot |G_p| - |G_p|$ is at most $(k - 1) \cdot |G_q|$. 

Hence, when performing the pairwise comparison between the $i$th highest value bundle in $G_p$ and the $i$th highest value bundle in $G_q$, the value of the bundles in $G_p$ is always greater than or equal to the value of the corresponding bundles in $G_q$. Therefore, Inequality (\ref{lemmaineq1})
holds, implying the desired condition 
\[
\frac{u\left(\bigcup_{i \in G_p} A_i\right)}{|G_p|} \geq \frac{u\left(\bigcup_{j \in G_q} A_j \setminus \{o\}\right)}{|G_q|}
\]
for some item $o \in \bigcup_{j \in G_q} A_j$.
\end{proof}

\begin{proof}[Proof of Theorem \ref{the: idential_valuations}]
By Lemmas \ref{lem:staggered matching-1} and \ref{lem:staggered matching-2}, and the fact that each agent has the same valuation, it can be concluded that the final allocation is EF1+CGEQ1.
Next, let us consider the time complexity. 
Without loss of generality, we assume that $m \geq n$.
In the first part of our algorithm, it takes $O(m\log m)$ to sort the items, and then for each iteration, selecting an agent takes $O(n)$ time, and choosing their favorite item takes $O(m)$ time. 
Note that there are $O(\lceil \frac{m}{n} \rceil)$ iterations.
In the second part of our algorithm, there are $O(n)$ iterations in the while loop.
For each iteration, selecting the target group and agent takes $O(nk)$ time, and choosing the bundle with the highest value from the centralized allocator's perspective takes $O(n)$ time.
Therefore, the total running time of our algorithm is $O(m^2+m \log m + n^2k)$.
\end{proof}

\begin{figure}[tb]

    \centering
    \scalebox{1}{
    \begin{tikzpicture}
        \definecolor{gpcolor}{RGB}{173,216,230}
        \definecolor{gqcolor}{RGB}{255,182,193}
        \node at (0, 0)   (b1) [draw, fill=gqcolor, minimum size=0.75cm] {$B_{q}^1$};
        \node at (1.25, 0) (b2) [draw, fill=gpcolor, minimum size=0.75cm] {$B_{p}^1$};
        \node at (2.5, 0)   (b3) [draw, fill=gqcolor, minimum size=0.75cm] {$B_{q}^2$};
        \node at (3.75, 0) (b4) [draw, fill=gqcolor, minimum size=0.75cm] {$B_{q}^3$};
        \node at (5, 0)   (b5) [draw, fill=gpcolor, minimum size=0.75cm] {$B_{p}^2$};
        \node at (6.25, 0) (b6) [draw, fill=gqcolor, minimum size=0.75cm] {$B_{q}^4$};
        \node at (7.5, 0)   (b7) [draw, fill=gqcolor, minimum size=0.75cm] {$B_{q}^5$};
        \draw[thick, ->] (-0.5, -0.5) -- (8, -0.5);
        \node at (7.5, -0.65) {time};
    \end{tikzpicture}
    }
    \caption{
        Example where $|G_p| = 2 $ and $|G_q| = 5$, with $G_q$ receiving the first bundle during the reallocation process. We denote $B_{z}^{i}$ as the bundle received by $G_{z}$ in the $i$-th allocation.
        The red squares represent bundles received by $G_q$, and the blue squares represent bundles received by $G_p$.
    }

    \label{BqandBp}
\end{figure}
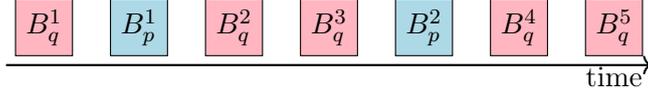

\section{Ordered Valuations}\label{sec:: ordered instance}

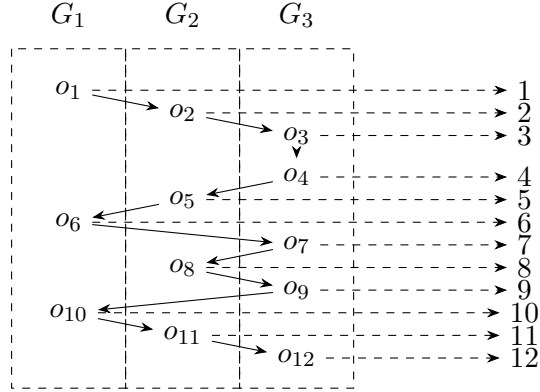
\begin{figure}[tb]
\begin{center}

\begin{tikzpicture}[->, >=Stealth, node distance=1.5cm]
    \node (o1) at (0, 0) {$o_1$};
    \node (o2) at (1.5, -0.3) {$o_2$};
    \node (o3) at (3, -0.6) {$o_3$};
    \node (o4) at (3, -1.15) {$o_4$};
    \node (o5) at (1.5, -1.45) {$o_5$};
    \node (o6) at (0, -1.75) {$o_6$};
    \node (o7) at (3, -2.05) {$o_7$};
    \node (o8) at (1.5, -2.35) {$o_8$};
    \node (o9) at (3, -2.65) {$o_9$};
    \node (o10) at (0, -2.95) {$o_{10}$};
    \node (o11) at (1.5, -3.25) {$o_{11}$};
    \node (o12) at (3, -3.55) {$o_{12}$};

    \draw[->] (o1) -- (o2);
    \draw[->] (o2) -- (o3);
    \draw[->] (o3) -- (o4);
    \draw[->] (o4) -- (o5);
    \draw[->] (o5) -- (o6);
    \draw[->] (o6) -- (o7);
    \draw[->] (o7) -- (o8);
    \draw[->] (o8) -- (o9);
    \draw[->] (o9) -- (o10);
    \draw[->] (o10) -- (o11);
    \draw[->] (o11) -- (o12);
    \draw[dashed] (-0.75, 0.55) rectangle (0.75, -3.95);
    \node (G1) at (0, 1){$G_1$};
    \draw[dashed] (0.75, 0.55) rectangle (2.25, -3.95); 
    \node (G2) at (1.5, 1){$G_2$};
    \draw[dashed] (2.25, 0.55) rectangle (3.75, -3.95);
    \node (G3) at (3, 1){$G_3$};

   \node (p1) at (6, 0) {$1$};
   \node (p2) at (6, -0.3) {$2$};
   \node (p3) at (6, -0.6) {$3$};
   \node (p4) at (6, -1.15) {$4$};
   \node (p5) at (6, -1.45) {$5$};
   \node (p6) at (6, -1.75) {$6$};
   \node (p7) at (6, -2.05) {$7$};
   \node (p8) at (6, -2.35) {$8$};
   \node (p9) at (6, -2.65) {$9$};
   \node (p10) at (6, -2.95) {$10$};
   \node (p11) at (6, -3.25) {$11$};
   \node (p12) at (6, -3.55) {$12$};
    \foreach \i in {1,...,12} {
        \draw[dashed] (o\i) -- (p\i);
    }
\end{tikzpicture}

\caption{An illustration of the allocation process of the first twelve items.
We use $\ell_i$ to denote the picking order of agent $i$ in her group.
Assume that there are three groups $G_1$, $G_2$, and $G_3$, where each group has $3$, $4$, and $5$ agents respectively. 
The arrow means the sequence of the allocation of these twelve items.
For example, in $G_1$, agent 1 with $\ell_1 = 1$ picks $o_1$, agent 6 with $\ell_6 = 2$ picks $o_6$, and agent 10 with $\ell_{10} = 3$ picks $o_{10}$.
Then, in the following iterations, if $G_1$ receives some item, the algorithm will follow the order to select the target agent.
} 
\label{fig: ideo_case}
\end{center}


\end{figure}

In this section, we consider the instance $\mathcal{I}$ with ordered valuations, where each agent $i \in \mathcal{N}$ and the centralized allocator share the same ranking or preference for all items. Specifically, $
 v_i (o_1) \geq \cdots \geq v_i(o_m)
$
and 
$
 u(o_1) \geq \cdots \geq u(o_m)
$.
In university settings, students typically agree on the relative order of course preferences (e.g., core before electives), but differ in the intensity of those preferences due to individual constraints such as major, prerequisites, or scheduling. EF1 remains a realistic fairness notion under such variations, while centralized rules determine access priorities. The challenge in this setting arises when certain items are valued oppositely by the agents and the centralized allocator. In such cases, ordered valuations may help us circumvent this issue. 


 Inspired by the algorithm that computes a weighted EF1 allocation in \cite{chakraborty2021weighted}, we propose the Synchronous Picking Sequence Algorithm, detailed in Algorithm \ref{alg: ido_case}. Our algorithm operates by allocating a set of items in several batches, where each batch has $n$ items. 
If there are fewer than $n$ items left, we can add some dummy items that have a zero value for the agents and the centralized allocator. 
For each item within a batch, there are two phases of allocation. 
In the first phase, the algorithm assigns the item to a group. 
In the second phase, the item is allocated to a specific agent within that group. 
Every agent receives exactly one item per batch. 
By structuring the allocation in this way, the algorithm mirrors a specialized round-robin algorithm.
Figure \ref{fig: ideo_case} illustrates the allocation process of the first batch (first twelve items) for an instance where there are three groups that have 3, 4, and 5 agents, respectively, as computed by Algorithm \ref{alg: ido_case}.

\begin{algorithm}[tb]
 \caption{Synchronous Picking Sequence (SPS) }\label{alg: ido_case}
 \label{alg: ido_case}
  \KwIn{An instance $\mathcal{I} = \langle \mathcal{O}, \mathcal{N}, \mathcal{G}, \boldsymbol{v}, u \rangle$ with ordered valuation functions}
  \KwOut{An EF1+CGEQ1 allocation $\mathcal{A}$}
Let $\mathcal{A} = (\emptyset, \ldots, \emptyset)$;

Add $n- (m \mod n)$ dummy items whose value is zero for the agents and the centralized allocator to $\mathcal{O}$;

 Assume that the groups are ordered in non-decreasing order of size, i.e., $\lvert G_1 \lvert \leq \ldots \leq \lvert G_k \rvert$;
 
Set $t_p \leftarrow 0$, $\forall p \in [k]$ and $\ell_{i} \leftarrow 0$, $\forall i \in \mathcal{N}$;

\While{$\mathcal{O} \neq \emptyset$}{

 Let $\mathcal{O}_n$ be the first $n$ items in $\mathcal{O}$ and $\mathcal{O} \leftarrow \mathcal{O} \setminus \mathcal{O}_n$;

 \While{$\mathcal{O}_n \neq \emptyset$}{

Let $o_{max} \in \arg \max_{ o \in \mathcal{O}_n}u(o)$;

 {\color{gray}       Phase 1: Decide which group gets this item      }
 
\If{$\exists t_{p} = 0 $}{
       $p^* \leftarrow \min \{p \vert p \in [k] ~\textit{and}~ t_p = 0\}$;\\   
     }
     \Else{

     $p^* \leftarrow \arg \min_{p\in [k]}{\frac{t_p}{|G_p|}}$ (breaking ties by selecting $p$ that reaches $\min_{p\in [k]}{\frac{t_p}{|G_p|}}$ the latest);\\ 
     }

 {\color{gray}   Phase 2: Decide which agent picks this item  }
 
\If{ $\exists i^* \in G_{p^*}$ such that $A_{i^*} = \emptyset$  }{

$A_{i^*} \leftarrow \{o_{max}\} $ and $\ell_{i^*} \leftarrow t_{p^*}$;

}
\Else{

Find the agent $i^*$ whose label $\ell_{i^*}$ equals  $t_{p^*} \mod \lvert G_{p} \rvert$, and $A_{i^*} \leftarrow A_{i^*} \cup \{o_{max}\}$;

}

$t_{p^*} \leftarrow t_{p^*}+1$ and $\mathcal{O}_n \leftarrow \mathcal{O}_n \setminus \{o_{max}\}$;

}
}

\Return $\mathcal{A}$ 
\end{algorithm}


\begin{theorem}
\label{the: ordered_instance}
    Given an instance with ordered valuations, the Synchronous Picking Sequence (Algorithm \ref{alg: ido_case}) computes an EF1+CGEQ1 allocation in polynomial time.
\end{theorem}

\begin{proof}
 For the centralized allocator, the notion of CGEQ1 can be seen as the variant of the notion of weighted EF1. In \cite{chakraborty2021weighted}, they showed that a weighted EF1 allocation can be computed in polynomial time, so it is not hard to see that the final allocation is CGEQ1.  
 For agents, the whole algorithm has two phases.
 In the first phase, each agent is relabelled in her group.
 In the second phase, if one group receives one item, following the rule of the label, the agent who has the smallest number picks this item.
 If there is a tie, the agent with the lowest index is selected.
 Besides that, each agent and the centralized allocator have the same preference, which means the item chosen by the centralized allocator in each iteration is also every agent's favorite item.
 Therefore, the whole allocation process for agents can be regarded as the round-robin protocol, and the final allocation is EF1 to agents.
\end{proof}

\section{Binary Allocator Valuations}
\label{sec: binary_allocator}

\begin{algorithm}[!tb]
 \caption{Clustering-Based Dual-Flow Picking (CD$^2$P)}
 \label{alg: Binary_allocator}
  \KwIn{An instance $\mathcal{I} = \langle \mathcal{O}, \mathcal{N}, \mathcal{G}, \boldsymbol{v}, u \rangle$ with binary allocator valuation function}
  \KwOut{An EF1+CGEQ1 allocation $\mathcal{A}$}
Let $\mathcal{A} = (\emptyset, \ldots, \emptyset)$;


Denote the collection of items with $u (o) = 1$ as $\mathcal{O}^1$ and the collection of items with $u (o)=0$ as $\mathcal{O}^2$;
 
Set $t_p \leftarrow 0$, $\forall p \in [k]$ and $\ell_{i} \leftarrow 0$, $\forall i \in \mathcal{N}$;\\
Set $t \leftarrow 0$;


\While{$\mathcal{O}^1 \neq \emptyset$}{




\If{ $\exists i$ such that $A_{i} = \emptyset$  }{
\If{$\exists t_{p} = 0 $}{
       $p^* \leftarrow \min \{p \vert p \in [k] ~\textit{and}~ t_p = 0\}$;\\   
     }
     \Else{

     $p^* \leftarrow \arg \min_{p\in [k]}{\frac{t_p}{|G_p|}}$ (breaking ties by selecting $p$ that reaches $\min_{p\in [k]}{\frac{t_p}{|G_p|}}$ the latest);\\ 
     }


Find agent $i^* \in G_{p^*}$ with $A_{i^*} = \emptyset$;\\
$A_{i^*} \leftarrow \{o_{max}\} $ \text{where} $o_{max} \in \arg\max_{ o \in \mathcal{O}^1}v_{i^*}(o)$;\\
$\ell_{i^*} \leftarrow t$;\\
 $t_{p^*}\gets t_{p^*}+1$;\\ 
}


\Else{

Find agent $i^*$ whose label $\ell_{i^*}$ equals  $t$\\
$A_{i^*} \leftarrow A_{i^*} \cup \{o_{max}\}$ \text{where} $o_{max} \in \arg \max_{ o \in \mathcal{O}^1}v_{i^*}(o)$;

}
$t\leftarrow (t+1)  \mod  n $;\\
$\mathcal{O}^1 \leftarrow \mathcal{O}^1 \setminus \{o_{max}\}$;

}

$t \gets n - 1;$\\


\While{$\mathcal{O}^2 \neq \emptyset$}{

 Find agent $i^*$ with $\ell_{i^*}$ equals $t$;\\
     $A_{i^*} \leftarrow A_{i^*} \cup \{o_{max}\}$ \text{where} $o_{max} \in \arg \max_{ o \in \mathcal{O}^2}v_{i^*}(o)$;\\
     $\mathcal{O}^2 \leftarrow \mathcal{O}^2 \setminus \{o_{max}\}$;\\
$t \gets (t - 1) \mod n$;
 }

\Return $\mathcal{A}$ 
\end{algorithm}

In this section, we consider the instance $\mathcal{I}$ where the centralized allocator has a binary valuation function, i.e., for each item $o$, either $u(o) = 0$ or $u(o) = 1$ holds.
We show that an EF1+CGEQ1 allocation always exists, which can be computed by the Dual-Flow Picking Algorithm (CD$^2$P, Algorithm \ref{alg: Binary_allocator}) in polynomial time.

We cluster the set of considered items $\mathcal{O}$ into two groups: $\mathcal{O}^{1}$, where the value of each item is 1 from the allocator's perspective, and $\mathcal{O}^2$, where the value of each item is 0 from the allocator's perspective. 
The idea behind Algorithm \ref{alg: Binary_allocator} is as follows: First, we figure out a special order for agents. This order is designed so that when we allocate items in $\mathcal{O}^{1}$ using a round-robin approach, it satisfies the CGEQ1 property. Once the allocation of $\mathcal{O}^1$ is completed, we proceed to allocate items in $\mathcal{O}^2$. 
Note that the allocation of items in $\mathcal{O}^{2}$ does not affect the CGEQ1 property.
For this allocation, we use the reverse of the previously determined sequence and perform another round-robin distribution based on the agents' valuation functions.

In vaccine or emergency-supply allocation, agencies first distribute critical items (allocator values $1$) before non-critical ($0$). This motivates our $\{0,1\}$ allocator-value case and the rule ``allocate $\mathcal{O}^{1}$ before $\mathcal{O}^{2}$," mirroring public-health prioritization under scarcity \cite{faden2020sage}.

\begin{theorem}
\label{the: binary_ef1_CGEQ1}
    Given an instance where the centralized allocator has the binary valuation function, the Group-Decided Round Robin Algorithm (Algorithm \ref{alg: Binary_allocator}) computes an EF1+CGEQ1 allocation in polynomial time.
\end{theorem}

We divide the proof of Theorem \ref{the: binary_ef1_CGEQ1} into two parts, corresponding to the two properties of the allocation computed by Algorithm \ref{alg: Binary_allocator}: CGEQ1 and EF1. These properties are established in the following lemmas.

\begin{lemma}
\label{lem: binary_CGEQ1}
    The output allocation in Algorithm \ref{alg: Binary_allocator} is CGEQ1.
\end{lemma}
\begin{proof}
    It suffices to show that at any point during the allocation process of the first bundle, the partial allocation is CGEQ1.
    We prove the statement by induction. Before allocating any item, the allocation is trivially CGEQ1. Fix two groups $p$ and $q$. Suppose that after allocating $k$ items, $G_p$ and $G_q$ receive $c_1$ and $c_2$ items, respectively, and $c_1 / |G_p| \leq c_2 / |G_q|$. If the $(k+1)$-th item is not allocated to $G_p$ or $G_q$, then $G_p$ and $G_q$ will not envy each other. Otherwise, the $(k+1)$-th item must go to $G_p$. Thus, we have $c_1' / |G_p| > c_1 / |G_p| \geq (c_2 - 1) / |G_q|$, and $c_2 / |G_q| \geq c_1 / |G_p| = (c_1' - 1) / |G_p|$.
\end{proof}

\begin{lemma}
\label{lem: binary_ef1}
The allocation computed by Algorithm \ref{alg: Binary_allocator} is EF1.
\end{lemma}

\begin{proof}
We assume that $|\mathcal{O}^1| = k_1 n$ and $|\mathcal{O}^2| = k_2 n$, as we can always achieve this by adding dummy items with value $0$ from all agents' perspectives.

Fix two agents $i$ and $j$ with $i < j$. Denote $i$'s items as $o_{i, 1}, \dots, o_{i, k_1 + k_2}$ and $j$'s items as $o_{j, 1}, \dots, o_{j, k_1 + k_2}$. We have $v_i(o_{i, k}) \geq v_i(o_{j, k})$ for $1 \leq k \leq k_1$ and $v_j(o_{j, k}) \geq v_j(o_{i, k+1})$ for $1 \leq k < k_1$. Similarly, we have $v_j(o_{j, k}) \geq v_j(o_{i, k})$ for $k_1 + 1 \leq k \leq k_1 + k_2$ and $v_i(o_{i, k}) \geq v_i(o_{j, k+1})$ for $k_1 + 1 \leq k < k_1 + k_2$.

Thus, we have
$$
\begin{aligned}
\sum_{k = 1}^{k_1 +k_2} v_i(o_{i , k})  &\geq \sum_{k = 1}^{k_1} v_i(o_{i , k}) + \sum_{k = k_1 + 1}^{k_1 +k_2 - 1} v_i(o_{i , k}) \\
& \geq \sum_{k = 1}^{k_1} v_i(o_{j, k}) + \sum_{k = k_1 + 2}^{k_1 + k_2} v_i(o_{j, k}) \\
& = \left(\sum_{k = 1}^{k_1 + k_2} v_i(o_{j, k})\right) - v_i(o_{j, k_1 + 1}),
\end{aligned}
$$
and
$$
\begin{aligned}
\sum_{k = 1}^{k_1 +k_2} v_j(o_{j , k})  &\geq \sum_{k = 1}^{k_1 - 1} v_j(o_{j , k}) + \sum_{k = k_1 + 1}^{k_1 + k_2} v_j(o_{j , k}) \\
& \geq \sum_{k = 2}^{k_1} v_j(o_{i, k}) + \sum_{k = k_1 + 1}^{k_1 + k_2} v_j(o_{i, k}) \\
&= \left(\sum_{k = 1}^{k_1 + k_2} v_j(o_{i, k})\right) - v_j(o_{i, 1}), 
\end{aligned}
$$
implying that agents $i$ and $j$ will not envy each other up to one item.
\end{proof}

\section{Centralized Group Maximin Share}
\label{sec: efficiency}
In the previous discussion, the allocator achieves group-level fairness among agents through an additional fairness requirement (CGEQ1). Now, we shift our focus to optimizing group-level fairness objectives directly, aiming to achieve fairness from a centralized perspective while still maintaining EF1 for the agents. This can be understood as the allocator striving to find the ``best" fair allocation.


In this case, the utilitarian social welfare ($\sum_{i=1}^{n} u(A_i)$) is not suitable to be the optimization objective since it not only remains invariant regardless of the allocation computed but also fails to reflect group-level fairness. Instead, we focus on the share-based fairness objective from the centralized allocator's perspective, which is called centralized group maximin share (CGMMS).
This definition is motivated by the well-studied notion---maximin share fairness (MMS) \citep{budish2011combinatorial}.
Our main goal is to find an allocation satisfying CGMMS from the centralized allocator's perspective and EF1 from the agents' perspective simultaneously.

\begin{definition}[Centralized Group Maximin Share]
Let $\mathcal{O}$ be the set of items and  $\Pi_{n}(\mathcal{O})$ be the set of $n$-partitions of $\mathcal{O}$ (which may be subject to some constraints).
The centralized group maximin share $\mathsf{CGMMS}$ is defined as:
\[
\mathsf{CGMMS}= \max_{\mathcal{A} \in \Pi_{n}(\mathcal{O})} \min_{G_p \in \mathcal{G}}
\frac{u(\cup_{i \in G_p}A_i)}{|G_p|}.
\]
An allocation $\mathcal{A}$ is centralized group maximin share fair (CGMMS) if $\min_{G_p \in \mathcal{G}} \frac{u(\cup_{i \in G_p}A_i)}{|G_p|}= \mathsf{CGMMS}$ holds.  
\end{definition}

First, regarding the existence of EF1+CGMMS, we give the following proposition to show that computing such an allocation is NP-hard.   
\begin{proposition} \label{prop:CGMMS+EF1}
    Computing a CGMMS allocation subject to EF1 for agents is strongly NP-hard.
\end{proposition}

\begin{proof}
It suffices to show that computing a CGMMS allocation is NP-hard. We reduce from the classic 3-partition problem. Let $I = \{v_1, v_2, \dots, v_{3n}\}$ be an instance of 3-partition, where the goal is to decide whether $I$ can be split into $n$ subsets with equal total value.

We construct a corresponding instance of the fair allocation problem with $3n$ items and $n$ agents, each from a different group. Each item in the partition instance corresponds to an item in the allocation problem, and the centralized allocator assigns value $v_i$ to item $i$. 

In this case, if we could compute such an allocation efficiently, we could decide whether the original 3-partition instance is a yes-instance by checking whether all group shares are equal. Therefore, computing a CGMMS allocation is NP-hard. 
\end{proof}




Then, we notice that for agents with identical valuation functions or instances with ordered valuations, computing an EF1+CGMMS allocation is still strongly NP-hard. However, when the centralized allocator has a binary valuation function, it can be solved efficiently.

\begin{theorem} \label{thm: CGMMS}
    When the centralized allocator has a binary valuation function, computing a CGMMS allocation subject to EF1 for agents can be achieved in polynomial time.
\end{theorem}

\begin{proof}
    We first notice that $\mathsf{CGMMS}$ (without considering EF1) can be computed in polynomial time. Since each item's value (from the allocator's perspective) is either 1 or 0, there are only polynomially many possible values for $u(\cup_{i \in G_p}A_i)/|G_p|$. Denote this set by $S$. Thus, we can enumerate values $x \in S$ and check whether $\mathsf{CGMMS} \geq x$ holds. The constraints are that each group should receive at least a certain number of items valued at 1, and the total number of items with value 1 should be sufficient to meet the requirements of all agents.

Let $r = \lfloor \mathsf{CGMMS} \rfloor$. Define $\mathcal{O}^1$ (resp. $\mathcal{O}^2$) as the set of items valued at 1 (resp. 0) from the allocator's perspective. There exists an allocation where each agent receives $r$ items, and additionally, some agents in each group may receive an extra item to ensure the allocation achieves $\mathsf{CGMMS}$. For these agents receiving extra items, we label them sequentially as $1, 2, \dots, (|\mathcal{O}^1| - n \cdot r)$. For the remaining agents, we label them as $(|\mathcal{O}^1| - n \cdot r + 1), \dots, n$.

Next, we show that there is an EF1 allocation that achieves $\mathsf{CGMMS}$. We apply a method similar to Algorithm \ref{alg: Binary_allocator}. Using the computed order, we perform a forward round-robin allocation of $\mathcal{O}^1$, followed by a reverse round-robin allocation of $\mathcal{O}^2$. By Lemma \ref{lem: binary_ef1}, the allocation computed is EF1, and it also achieves $\mathsf{CGMMS}$.
\end{proof}

\section{Conclusion}
\label{sec::conlusion}
In this paper, we study the fair division of indivisible items from the perspectives of agents and a centralized allocator. 
We propose to use EF1 and CGEQ1 to measure the fairness from the agents' and the centralized allocator's perspectives, respectively, and aim to compute allocations that satisfy EF1 and CGEQ1 simultaneously. 
We show that EF1+CGEQ1 allocations always exist for different classes of agents' and the centralized allocator's valuation functions, which can be computed in polynomial time. As for optimizing group-level fairness objectives, we show that, in general, finding a CGMMS allocation is hard, but an EF1+CGMMS allocation can be computed within polynomial time when the centralized allocator has a binary valuation function.

For future work, a natural direction is to determine whether an allocation satisfying the above two fairness notions exists in more general settings. \emph{We have searched for a non-existence counterexample with the aid of computer programs, but 
it seems to be hard to find such an instance.} Further, we can explore the setting where the agents and the centralized allocator have beyond additive valuation functions like submodular or subadditive valuation functions and design algorithms that efficiently return EF1+CGEQ1 allocations. For the group-level fairness objective, we can approximately optimize $\mathsf{CGMMS}$ subject to EF1 for agents.






\bibliographystyle{plainnat}
\bibliography{ref}
\end{document}